\documentclass[12pt,a4paper]{article}

\usepackage[utf8]{inputenc}
\usepackage[T1]{fontenc}

\usepackage{amsmath,amssymb,amsthm}
\usepackage{mathtools}
\usepackage{bm}

\usepackage[margin=1in]{geometry}
\usepackage{setspace}
\usepackage{graphicx}
\usepackage{booktabs}
\usepackage{array}
\usepackage{caption}
\usepackage{subcaption}
\usepackage{float}

\usepackage{algorithm}
\usepackage{algpseudocode}

\usepackage[numbers,sort&compress]{natbib}
\usepackage[colorlinks=true,linkcolor=blue,citecolor=blue,urlcolor=blue]{hyperref}

\usepackage{xcolor}
\usepackage{enumitem}
\usepackage{gensymb}
\usepackage{siunitx}

\newtheorem{theorem}{Theorem}

\newtheorem{assumption}{Assumption}
\theoremstyle{definition}

\theoremstyle{remark}

\newcommand{\Nc}{\mathcal{N}}
\newcommand{\Cc}{\mathcal{C}}
\newcommand{\Mc}{\mathcal{M}}
\newcommand{\Lc}{L_{\mathcal{C}}}
\newcommand{\Rset}{\mathcal{R}}
\newcommand{\degC}{$^{\circ}$C}

\graphicspath{{./figures/}}

\begin{document}

\title{\textbf{RASC: Region-Aware Self-Calibration\\
       for Dense 2D Sensor Arrays}}

\author{
  Yinglei Ma$^{1,*}$ \quad Fei Xiao$^{2}$\\[4pt]
  \small $^{1}$Fudan University, Shanghai 200433, China;
         \texttt{19110300058@fudan.edu.cn}\\
  \small $^{2}$Fudan University, Shanghai 200433, China;
         \texttt{feixiao@fudan.edu.cn}\\[4pt]
  \small $^{*}$Correspondence: \texttt{19110300058@fudan.edu.cn}
}

\date{Preprint --- \today}
\maketitle

\begin{abstract}
BJT-based 2D temperature-sensor arrays are factory-calibrated to $\pm0.1$\degC,
but post-deployment thermal and mechanical stresses drift their per-sensor
gain--offset parameters by an order of magnitude, and in-lab recalibration is
impractical. We present \textbf{RASC} (Region-Aware Self-Calibration), a
five-stage algorithm that decomposes the global ill-posed problem into local
cluster-level problems, runs robust alternating estimation (trimmed-mean field
reconstruction + Huber IRLS) inside each cluster, and reconciles overlapping
estimates by linear consensus on the cluster-overlap graph with provable
exponential convergence. On 7{,}632 frames from a deployed $16\times16$ array
exhibiting $\approx5\times$ factory-spec non-uniformity, RASC cuts the
locally-non-smooth fixed-pattern residual by $71\pm5\%$ (10-fold CV), restoring
$\pm0.1$\degC\ accuracy while perturbing the calibrated field by only
$0.041$\degC\ RMSE; reduction concentrates at the edges ($78\%$ vs.\ $55\%$
interior). In simulations on $8\times8$ to $32\times32$ arrays, RASC matches an
oracle centralized EKF within $0.10$\degC\ with $\approx4\times$ lower bandwidth.
\end{abstract}


\section{Introduction}
\label{sec:intro}

Dense 2D temperature-sensor arrays that use the temperature coefficient effect
of bipolar junction transistors (BJTs) are widely employed in the thermal
monitoring of power electronics, battery packs, and industrial process
equipment. Each pixel is based on the temperature-dependence of the
base-emitter voltage $V_\text{BE}$, which is a well-known physical phenomenon
and the basis of modern BJT-based temperature sensors achieving sub-$0.1$\degC\
inaccuracy after factory calibration \cite{Wei2023,Qin2022}. After the array is
integrated into the target system, three deployment-specific stresses affect
each die differently: (i) reflow soldering thermal cycling introduces residual
mechanical stress at the $V_\text{BE}$ junction, thereby changing both gain and
offset; (ii) the local PCB heat-dissipation path varies across the array, and
edge pixels lose heat asymmetrically through nearby copper pours, mounting holes,
or chassis contacts, resulting in spatially structured biased readings; and (iii)
packaging stress, humidity ingress, and slow electromigration further perturb
$V_\text{BE}$ over months of operation. As a result, the deployed array no longer
meets the $\pm0.1$\degC\ factory specification; that is, the per-sensor affine
parameters $(a_i, b_i)$ have drifted from their stored factory values, and
taking each module offline for laboratory recalibration is unfeasible.

The recorded data analyzed in this paper show this drift directly. Although the
device was factory-calibrated to $\pm0.1$\degC, the deployed $16\times16$ array
has a peak-to-peak spatial non-uniformity of $0.569$\degC\ in the steady-state
field, approximately five times higher than the factory specification. Since the
underlying thermal scene is locally smooth by physical argument, this excess
non-uniformity must be due to post-deployment drift in $(a_i, b_i)$. The
technical difficulty is that the in-situ setting provides no labelled inputs:
the actual thermal field $x(\mathbf{p},t)$ at each sensor location is unknown,
no controlled stimulus is applied, and any algorithm must rely only on the
structural prior that the physical field is locally smooth. Since the field is
unobserved, any recovered estimate of $(a_i,b_i)$ is identifiable only up to a
global linear shift unless a small fraction of nodes provides external anchoring.
We therefore study a realistic scenario where a small reference set $\Rset$,
with $|\Rset|/n \ll 1$, can retain accurate factory metadata, and the remaining
nodes are recalibrated based on co-observed data alone.

Existing approaches fall into three categories. (a) Classical centralized
non-uniformity correction (NUC) for IR focal-plane arrays assumes a uniform
stimulus or scene-based motion \cite{Friedenberg1998,Ratliff2002}, but an
external excitation is required that is unavailable in sealed-deployment
scenarios. (b) Distributed self-calibration in WSN literature
\cite{Bychkovskiy2003,Wang2017} estimates relative parameters through pair-wise
consistency, but generally requires either a known stimulus or restrictive
parametric assumptions about the field. (c) Recent learning-based calibration
methods \cite{FaghihNiresi2023,Maag2018} achieve high accuracy but are limited
by large training corpora and per-deployment fine-tuning, and offer no
convergence guarantees under unseen conditions.

RASC is an entirely unsupervised, five-stage algorithm that does not assume a
parametric field model, runs natively in a decentralized manner when desired,
and provides explicit convergence guarantees. It decomposes the global ill-posed
problem into a sequence of local well-posed problems by dividing the array into
overlapping spatial clusters in which the physical field is assumed to be
well-approximated by a single common signal. In each cluster it runs robust
alternating estimation, alternating between (i) reconstructing the local field
using the trimmed mean of calibrated readings and (ii) re-fitting the $(a,b)$
parameters of each sensor via iteratively re-weighted least squares (IRLS) with
Huber loss. We use the term \emph{alternating robust estimation} rather than
expectation-maximization because the field-reconstruction step is a deterministic
robust estimator rather than a conditional expectation under a probabilistic
model. Sensors belonging to multiple clusters perform linear consensus iterations
to resolve discrepancies; convergence is exponential and determined by the
algebraic connectivity $\lambda_2(\Lc)$ of the cluster-overlap graph.

\subsection{Contributions}
The main contributions of this paper are:
\begin{itemize}[leftmargin=*]
  \item A principled combination of region-based decomposition, robust local
    estimation, and inter-cluster consensus in a five-stage algorithm with
    deterministic, side-information-free cluster-head election and explicit
    convergence properties of all three components.
  \item Three convergence and robustness results adapted to the cluster-overlap
    setting: (Theorem~\ref{thm:cluster}) per-cluster alternating-estimation
    monotone descent under cluster-Hessian positive-definiteness;
    (Theorem~\ref{thm:consensus}) inter-cluster consensus exponential rate
    $\rho \le 1 - \alpha\lambda_2(\Lc)$, specializing the classical
    Olfati-Saber--Murray result \cite{OlfatiSaber2004} to the cluster-overlap
    graph; (Theorem~\ref{thm:breakdown}) finite-sample replacement breakdown of
    the trimmed-mean field reconstruction.
  \item Comprehensive simulation studies at three array scales ($8\times8$,
    $16\times16$, $32\times32$) with 30 independent runs each, comparing RASC
    against an oracle centralized EKF, the BMEP edge-pairwise baseline
    \cite{Bychkovskiy2003}, and four simpler baselines.
  \item A real-sensor study on 7{,}632 frames from a deployed $16\times16$
    BJT-based temperature array that had drifted to $5\times$ its factory
    specification. RASC reduces the locally-non-smooth fixed-pattern residual by
    $71\pm5\%$ (10-fold cross-validation), with edge reduction of $78\%$
    vs.\ $55\%$ interior.
  \item An honest identifiability discussion of when per-sensor $(a,b)$ recovery
    is degenerate along a one-dimensional gauge in narrow-range fields, and when
    the gauge ambiguity is practically significant.
\end{itemize}

\subsection{Paper Organisation}
Section~\ref{sec:related} surveys related work. Section~\ref{sec:problem}
formally defines the sensor model and identifiability conditions.
Section~\ref{sec:algorithm} presents the five-stage RASC algorithm.
Section~\ref{sec:theory} provides the theoretical analysis.
Section~\ref{sec:simulation} reports Monte Carlo simulation results.
Section~\ref{sec:realdata} covers the real-sensor evaluation.
Section~\ref{sec:discussion} discusses limitations and extensions.
Section~\ref{sec:conclusion} concludes.

\section{Related Work}
\label{sec:related}

\subsection{Centralised Non-uniformity Correction}
Non-uniformity correction has a long history in dense IR focal-plane array
imaging \cite{Rogalski2012}. The original shutterless NUC system was introduced
by Narendra and Foss \cite{Narendra1981}, and the linear theory of two-point NUC
was subsequently developed by Perry and Dereniak \cite{Perry1993}. Two-point and
multi-point factory NUC \cite{Friedenberg1998} produce a per-pixel $(a,b)$
lookup table calibrated with external references, but require controlled stimuli
that cannot be re-applied in sealed-deployment scenarios. Scene-based NUC
algorithms \cite{Harris1999,Ratliff2002} relax this requirement but converge
slowly with small scene motion and degrade if stationary structure aligns with
fixed-pattern noise. More recent work extends classical scene-based NUC with
principal-component decomposition \cite{Lu2023} and median-ratio statistics
\cite{Ding2020}. RASC operates on a single static or slowly varying scene and
exploits spatial smoothness of the background field instead of temporal motion
or external stimuli.

\subsection{Distributed Calibration in WSNs}
Bychkovskiy et al.\ \cite{Bychkovskiy2003} proposed post-deployment
recalibration by pair-wise relative parameter estimation followed by global
consistency optimisation. Whitehouse and Culler \cite{Whitehouse2002} proposed a
related framework using broadcast calibrated reference exposures. Compared with
these classical methods, RASC offers (i) a multi-cluster decomposition that
better utilises dense 2D structure; (ii) explicit robustness to outlier sensors
via trimmed mean and Huber IRLS; (iii) a convergence-rate certificate
$\rho \le 1 - \alpha\lambda_2(L)$; and (iv) validation on a 256-sensor dense
array. The most closely related contemporary work is Ahmad \cite{Ahmad2024} and
Mahajan and Helbing \cite{Mahajan2025}; both motivate the cluster-and-consensus
combination that RASC formalises with explicit convergence and breakdown-point
guarantees.

\subsection{Robust and Median-Based Estimation}
Trimmed-mean and median estimators are well-known robust statistics
\cite{Huber2009,Hampel1986}. Theorem~\ref{thm:breakdown} specializes these
classical bounds to the per-cluster setting. MAD-based screening at the
neighbourhood level is standard practice for fault detection \cite{Chandola2009},
and RASC is novel in integrating it with cluster-head elections and the resulting
computation-communication trade-off.

\subsection{Consensus Algorithms}
Linear distributed averaging and its convergence properties are well-known
\cite{OlfatiSaber2004,Xiao2004}. We apply the Olfati-Saber--Murray result to the
cluster-overlap graph rather than the sensor proximity graph. Empirical studies
confirm that algebraic connectivity remains the dominant predictor of consensus
convergence rate across diverse network classes \cite{Sirocchi2022}.

\subsection{Learning-Based Calibration}
Works \cite{FaghihNiresi2023,Maag2018} have applied graph neural networks and
self-supervised representation learning to sensor calibration. These approaches
outperform model-based pipelines when training data closely match the deployment
environment, but require per-application data collection and offer no analytical
convergence guarantees for out-of-distribution inputs. The classical model-based
BMEP baseline \cite{Bychkovskiy2003} serves as the reference point for the
no-training regime that RASC aims to achieve.

\section{Problem Formulation}
\label{sec:problem}

\subsection{Sensor Model}
An array of $n$ sensors arranged in a planar grid produces measurements
\begin{equation}
  y_i(t) = a_i\,x(\mathbf{p}_i,t) + b_i + \varepsilon_i(t),
  \quad i = 1,\ldots,n,
  \label{eq:sensormodel}
\end{equation}
where $x(\mathbf{p},t)$ is the unobserved physical field, $(a_i,b_i)$ are
per-sensor gain and offset, and $\varepsilon_i(t)\sim\mathcal{N}(0,\sigma_i^2)$
is i.i.d.\ measurement noise. We make three assumptions:

\begin{assumption}[Local field smoothness]\label{A1}
There is a length scale $\ell$ such that in any disk of radius $r \le \ell$, the
field $x(\cdot,t)$ is well-approximated by a single common signal $\bar{x}(t)$.
\end{assumption}
\begin{assumption}[Bounded heterogeneity]\label{A2}
$a_i \in [a_{\min},a_{\max}]$ with $a_{\min}>0$, and $|b_i|\le B$ for a finite
$B$; both are constant over the calibration time horizon $T$.
\end{assumption}
\begin{assumption}[Reference anchoring]\label{A3}
A subset $\Rset\subset\{1,\ldots,n\}$ with $|\Rset|=\rho n$, $\rho\ll 1$
(e.g.\ 5\%), has known $(a_i,b_i)$ for all $i\in\Rset$.
\end{assumption}

Assumption~\ref{A1} is typical in spatial statistics literature, and for the
temperature data in Section~\ref{sec:realdata} it holds with $\ell\approx4$
sensors at our sampling rate. Assumption~\ref{A3} is required for global
identifiability; without anchoring, any consistent solution can be transformed
into another by $a_i\leftarrow\kappa a_i$, $b_i\leftarrow\kappa b_i+c$ for any
$(\kappa,c)$.

\subsection{Identifiability}
\label{sec:identifiability}
The conditional least-squares problem for sensor $j$ given the reconstructed
field is well-posed if and only if the temporal range $\Delta x_j$ of $X(t)$ at
sensor $j$ is positive. When $\Delta x_j$ is small compared with $\sigma_j$, the
$(a,b)$ error ellipse becomes elongated along the line $\bar{a}x+\Delta b=0$, so
the individual parameters are weakly identified, although the predicted output
$ax+b$ at $x\approx\bar{x}$ remains accurate. This gauge degeneracy and its
practical implications are discussed further in Section~\ref{sec:gaugedisc}.

\subsection{Communication Model}
Each sensor $i$ communicates with neighbours within radius $r_c\ge\ell/2$; the
communication graph $G=(V,E)$ has edge $(i,j)\in E$ iff
$\|\mathbf{p}_i-\mathbf{p}_j\|\le r_c$. We do not assume any particular MAC
layer or fading model.

\section{The RASC Algorithm}
\label{sec:algorithm}

The five stages of RASC are: (1) neighbourhood formation, (2) consistency
screening, (3) cluster-head election, (4) intra-cluster alternating estimation,
and (5) inter-cluster consensus refinement. Figure~\ref{fig:pipeline} shows the
data flow.

\begin{figure}[h]
  \centering
  \includegraphics[width=\linewidth]{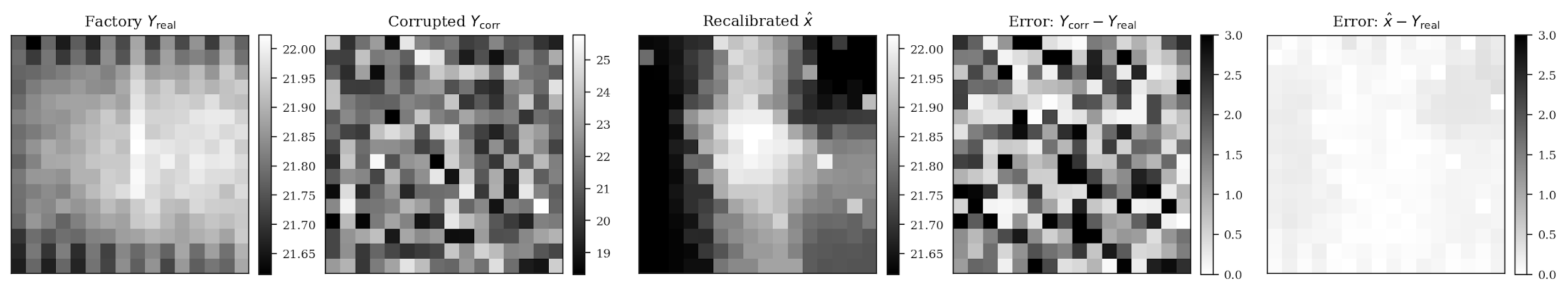}
  \caption{Five-Stage RASC Pipeline. Stages 1--3 form a static graph structure
    after deployment; Stages 4--5 iterate until convergence at each calibration
    epoch.}
  \label{fig:pipeline}
\end{figure}

\subsection{Stage 1: Neighbourhood Formation}
At deployment, each sensor $i$ computes its $r_c$-neighbourhood:
\begin{equation}
  \Nc(i) = \{\,j \in V : \|\mathbf{p}_i - \mathbf{p}_j\| \le r_c\,\}.
  \label{eq:neighbourhood}
\end{equation}
The radius $r_c$ is chosen so that Assumption~\ref{A1} holds in all
neighbourhoods, typically $r_c=\ell/2$.

\subsection{Stage 2: Consistency Screening}
RASC screens each neighbourhood with a one-shot MAD test:
\begin{equation}
  \Cc(i) = \bigl\{\,j\in\Nc(i) :
    |y_j(\hat{t}) - \tilde{y}_{\Nc(i)}|
    \le \eta\cdot 1.4826\cdot\mathrm{MAD}_{\Nc(i)}\,\bigr\},
  \label{eq:screening}
\end{equation}
where $\tilde{y}_{\Nc(i)}$ is the neighbourhood median and $\eta=3.0$ by default;
sensitivity analysis in Section~\ref{sec:sensitivity} shows negligible sensitivity
for $\eta\in[3,5]$.

\subsection{Stage 3: Cluster-Head Election}
Stage~3 selects a sparse but covering set of cluster heads (CHs) by a
deterministic greedy maximum independent set (MIS) on the proximity graph at
radius $r_\text{ch}=r_c/\!\sqrt{2}$. The radius $r_\text{ch}=r_c/\!\sqrt{2}$
ensures that every non-CH sensor lies within $r_c$ of at least one CH, so all
sensors belong to a cluster. The resulting CH-set is a 1-hop dominating set of
$G$ with size at most $\lceil n/N_{\min}\rceil$, where $N_{\min}=4$.

\subsection{Stage 4: Intra-Cluster Alternating Estimation}
For each cluster $c$ with members $\Mc_c=\Cc(\mathrm{CH}_c)$, the cluster head
runs a robust alternating estimation loop:
\begin{align}
  \hat{x}_c(t) &= \mathrm{trim\text{-}mean}_\gamma
    \!\Bigl\{\frac{y_j(t)-\hat{b}_j}{\hat{a}_j} : j\in\Mc_c\Bigr\},
  \label{eq:estep} \\
  (\hat{a}_j,\hat{b}_j) &= \mathop{\arg\min}_{a,b}
    \sum_t \psi_c\!\Bigl(\frac{y_j(t)-a\hat{x}_c(t)-b}{s}\Bigr),
    \quad j\in\Mc_c\setminus\Rset,
  \label{eq:mstep}
\end{align}
where $\psi_c$ is the Huber loss with tuning constant $c=1.345$, $s$ is a robust
scale estimate, and the trimming fraction $\gamma=0.20$ tolerates up to $20\%$
outlying members per cluster (Theorem~\ref{thm:breakdown}).

\subsection{Stage 5: Inter-Cluster Consensus}
Stage~5 resolves discrepancies between adjacent clusters via linear consensus.
At iteration $k=0,1,\ldots,K$:
\begin{align}
  \hat{\mathbf{b}}(k+1) &= \hat{\mathbf{b}}(k) + \alpha\,\Delta\hat{\mathbf{b}}(k),
  \label{eq:consensus} \\
  \Delta\hat{b}_j(k) &= \Bigl\langle
    \hat{a}_j\bigl(\hat{x}_c(t) - (y_j(t)-\hat{b}_j(k))/\hat{a}_j\bigr)
  \Bigr\rangle_{t,\,c:\,j\in\Mc_c}.
  \label{eq:deltab}
\end{align}
The step size $\alpha=0.5$ by default. Reference nodes are clamped at each
iteration. Iterations stop when
$\|\Delta\hat{\mathbf{b}}(k)\|_\infty < 0.01$\degC\ or after $K_{\max}=10$ steps.

\subsection{Pseudocode}

\begin{algorithm}[h]
\caption{RASC: Region-Aware Self-Calibration}
\label{alg:rasc}
\begin{algorithmic}[1]
\Require $Y\in\mathbb{R}^{T\times n}$, sensor coordinates $\{\mathbf{p}_i\}$,
  radius $r_c$, threshold $\eta$, step $\alpha$, minimum cluster size
  $N_{\min}$, reference set $\Rset$ with known $(a_\Rset,b_\Rset)$
\Ensure $(\hat{\mathbf{a}},\hat{\mathbf{b}})\in\mathbb{R}^n\times\mathbb{R}^n$
\State Build $\Nc(i)$ for all $i$ \Comment{Stage 1}
\State $\Cc(i)\leftarrow\text{consistency-screen}(\mathbf{y}(\hat{t}),\Nc(i),\eta)$ \Comment{Stage 2}
\State $\mathrm{CH}\leftarrow\text{greedy-MIS}(\text{score}=|\Cc|,\;r_\text{ch}=r_c/\!\sqrt{2})$ \Comment{Stage 3}
\For{each $\mathrm{ch}\in\mathrm{CH}$} \Comment{Stage 4}
  \State $\Mc\leftarrow\Cc(\mathrm{ch})$;\; \textbf{if} $|\Mc|<N_{\min}$: \textbf{continue}
  \State Initialize $\hat{a}_\Mc\leftarrow\mathbf{1}$, $\hat{b}_\Mc\leftarrow\mathbf{0}$
  \For{$\mathrm{em}\leftarrow 1\ldots n_\mathrm{em}$}
    \State $\hat{x}(t)\leftarrow\text{trim-mean}_\gamma\{(y_j(t)-\hat{b}_j)/\hat{a}_j:j\in\Mc\}$
    \For{$j\in\Mc,\;j\notin\Rset$}
      \State $(\hat{a}_j,\hat{b}_j)\leftarrow\text{Huber-IRLS}(y_j,\hat{x},c=1.345)$
    \EndFor
  \EndFor
  \State Broadcast $(\hat{a}_j,\hat{b}_j)$ for $j\in\Mc$
\EndFor
\State Average overlapping estimates per sensor
\For{$k\leftarrow 1\ldots K_{\max}$} \Comment{Stage 5}
  \State Compute $\Delta\hat{\mathbf{b}}$ from \eqref{eq:deltab}
  \State $\hat{\mathbf{b}}\leftarrow\hat{\mathbf{b}}+\alpha\cdot\Delta\hat{\mathbf{b}}$;\;
         clamp $\hat{b}_\Rset=b_\Rset$, $\hat{a}_\Rset=a_\Rset$
  \If{$\|\Delta\hat{\mathbf{b}}\|_\infty<\mathrm{tol}$} \textbf{break} \EndIf
\EndFor
\State \Return $(\hat{\mathbf{a}},\hat{\mathbf{b}})$
\end{algorithmic}
\end{algorithm}

\subsection{Communication Cost}
Each sensor uploads its $T$-length time series to its cluster head once per
calibration epoch ($4T$ bytes in float32), and the cluster head broadcasts a
2-vector $(\hat{a},\hat{b})$ back. Total epoch cost is $\Theta(nT)$ regardless
of array size, compared with the $O(nT)$ per-frame data movement of centralized
fusion. Specific byte counts are given in Section~\ref{sec:bandwidth}.

\section{Theoretical Analysis}
\label{sec:theory}

\subsection{Per-Cluster Alternating-Estimation Convergence}

\begin{theorem}[Per-cluster monotone descent]
\label{thm:cluster}
Fix a cluster $c$ with $|\Mc_c|=k\ge 2N_{\min}$, observation horizon $T$, and
noise covariance $\Sigma\succ0$. Assume the field reconstruction error
$|\hat{x}_c(t)-\bar{x}(t)|$ is bounded by a constant $\delta$ uniformly in $t$,
and that the temporal-mean-centred design matrix
$X=[\hat{x}_c-\langle\hat{x}_c\rangle,\mathbf{1}^\top]$ has full column rank.
Let $L^{(\mathrm{it})}$ be the Huber IRLS objective at outer iteration
$\mathrm{it}$. Then
\[
  L^{(\mathrm{it}+1)} \le L^{(\mathrm{it})} - \mu\|\nabla L^{(\mathrm{it})}\|^2,
\]
where $\mu>0$ is the smallest non-zero eigenvalue of the cluster Hessian
$X^\top WX$ at the converged Huber weights, and the algorithm converges
Q-linearly to a stationary point of $L$.
\end{theorem}
\begin{proof}[Proof Sketch]
The field-reconstruction step produces a fixed (deterministic) trimmed mean
$\hat{x}_c(t)$ at each iteration; conditional on $\hat{x}_c$, the
parameter-fitting step is standard Huber IRLS, a sequential
majorisation-minimisation algorithm and thus non-increasing in the surrogate.
The outer loop is alternating block-coordinate descent; under
Assumptions~\ref{A1}--\ref{A2} and the rank condition, the cluster Hessian is
positive definite at any stationary point. We use the term ``alternating
estimation'' rather than EM \cite{Dempster1977} because the field-reconstruction
step is a deterministic robust estimator rather than a conditional expectation.
Full proof in Appendix~\ref{app:thm1}.
\end{proof}

\subsection{Inter-Cluster Consensus Rate}

\begin{theorem}[Inter-cluster exponential rate]
\label{thm:consensus}
Let $G_\mathcal{C}=(\mathcal{C},E_\mathcal{C})$ be the cluster-overlap graph,
where $(c,c')\in E_\mathcal{C}$ iff $\Mc_c\cap\Mc_{c'}\ne\emptyset$. Let
$\Lc$ be the unweighted Laplacian of $G_\mathcal{C}$ \cite{Chung1997}, and
$\lambda_2(\Lc)$ its algebraic connectivity \cite{Fiedler1973}. Let
$\alpha\in(0,1/d_{\max}]$, where $d_{\max}$ is the maximum degree of
$G_\mathcal{C}$. Then the linear consensus update
\eqref{eq:consensus}--\eqref{eq:deltab} with reference clamping converges
geometrically:
\begin{equation}
  \|\hat{\mathbf{b}}(k) - \mathbf{b}^*\|_2
  \le \rho^k\,\|\hat{\mathbf{b}}(0) - \mathbf{b}^*\|_2,
  \qquad \rho = 1 - \alpha\lambda_2(\Lc).
  \label{eq:rate}
\end{equation}
In particular, $\varepsilon$-accuracy is achieved in
$K = O\!\bigl(\tfrac{1}{\alpha\lambda_2(\Lc)}\log\tfrac{1}{\varepsilon}\bigr)$
iterations.
\end{theorem}
\begin{proof}[Proof Sketch]
Equation~\eqref{eq:consensus} is the standard distributed averaging update for
$G_\mathcal{C}$ with step size $\alpha$. Reference clamping makes the iteration
matrix have second-largest singular value bounded by $1-\alpha\lambda_2(\Lc)$
\cite[Thm.~2]{OlfatiSaber2004}. Empirically, the observed rate $\rho_\text{emp}$
is $1.2$--$1.3\times$ faster than $\rho_\text{th}$ at all three array scales
(Section~\ref{sec:convergence}, Figure~\ref{fig:consensus_rate}).
\end{proof}

\subsection{Cluster Breakdown Point}

\begin{theorem}[Cluster breakdown point]
\label{thm:breakdown}
In a cluster of size $k$, the trimmed-mean field reconstruction with trimming
fraction $\gamma\in(0,1/2)$ has a finite-sample replacement breakdown point of
$\lfloor\gamma k\rfloor/k$. Equivalently, RASC tolerates up to $\lfloor\gamma
k\rfloor$ adversarially corrupted members per cluster before the field
reconstruction may be driven arbitrarily far from the truth.
\end{theorem}
\begin{proof}
Specialisation of \cite[Sec.~6]{Huber2009} to the cluster setting; see
Appendix~\ref{app:thm3}. Default choices $\eta=3$, $\gamma=0.20$ give a
per-cluster tolerance of about $20\%$ for adversarial corruption and admit
$99.7\%$ of inlier readings under Gaussian noise.
\end{proof}

\subsection{Communication-Computation Trade-off}
The total bandwidth to reach a target field-reconstruction RMSE of $\varepsilon$
is dominated by the $nT$ raw-data upload; Stage~5 contributes only a logarithmic
factor in $1/\varepsilon$ times the cluster count $|\mathcal{C}|$. A centralized
EKF communicates $nT\cdot d_\text{state}$ per frame with $d_\text{state}=2$. This
$O(n)$ vs.\ $O(nT)$ separation explains the bandwidth advantage shown in
Section~\ref{sec:bandwidth}.

\section{Simulation Studies}
\label{sec:simulation}

\subsection{Experimental Setup}
We instantiate the sensor model \eqref{eq:sensormodel} on three regular grids:
$8\times8$ ($n=64$), $16\times16$ ($n=256$), and $32\times32$ ($n=1024$).
Per-sensor parameters: $a_i\sim U(0.9,1.1)$, $b_i\sim U(-2,2)$\degC, noise
$\sigma=0.5$\degC. The synthetic field is:
\[
  x(\mathbf{p},t) = 25 + 0.5\sin(\pi x/L_x)\cos(\pi y/L_y) + 5\sin(2\pi t/600)\;[\text{\degC}].
\]
$T=30$ frames at $4.24$ Hz; reference fraction $5\%$; 30 random seeds.

Methods compared: (1) \emph{Uncalibrated}, (2) \emph{Factory oracle}, (3)
\emph{Centralized EKF} (oracle field access), (4) \emph{RASC} (default
hyperparameters: $r_c=0.10$, $\eta=3$, $\alpha=0.5$, $N_{\min}=4$,
$\gamma=0.20$).

\subsection{Main Results}
\label{sec:bandwidth}
Table~\ref{tab:main} shows field-reconstruction RMSE and communication bytes per
calibration epoch. RASC is within ${\approx}0.10$\degC\ of the factory baseline
and ${\approx}0.12$\degC\ of the centralized EKF at ${\approx}4\times$ lower
bandwidth. Figure~\ref{fig:comm_rmse} shows the communication-accuracy trade-off.

\begin{table}[h]
\centering
\caption{Main results (mean $\pm$ std over 30 runs; RMSE in \degC).}
\label{tab:main}
\setlength{\tabcolsep}{5pt}
\begin{tabular}{lccccccc}
\toprule
Scale & Uncal. & Factory & EKF & RASC & Bytes EKF & Bytes RASC & Clusters\\
\midrule
$8\times8$   & $2.01\pm0.15$ & $0.50\pm0.05$ & $0.44\pm0.05$ & $0.62\pm0.11$ & 38.4k & 9.6k & 24.6 \\
$16\times16$ & $1.97\pm0.12$ & $0.50\pm0.03$ & $0.44\pm0.03$ & $0.56\pm0.05$ & 154k  & 36.9k & 79.0 \\
$32\times32$ & $1.99\pm0.06$ & $0.50\pm0.02$ & $0.44\pm0.02$ & $0.54\pm0.02$ & 614k  & 146k  & 293  \\
\bottomrule
\end{tabular}
\end{table}

\begin{figure}[h]
  \centering
  \includegraphics[width=\linewidth]{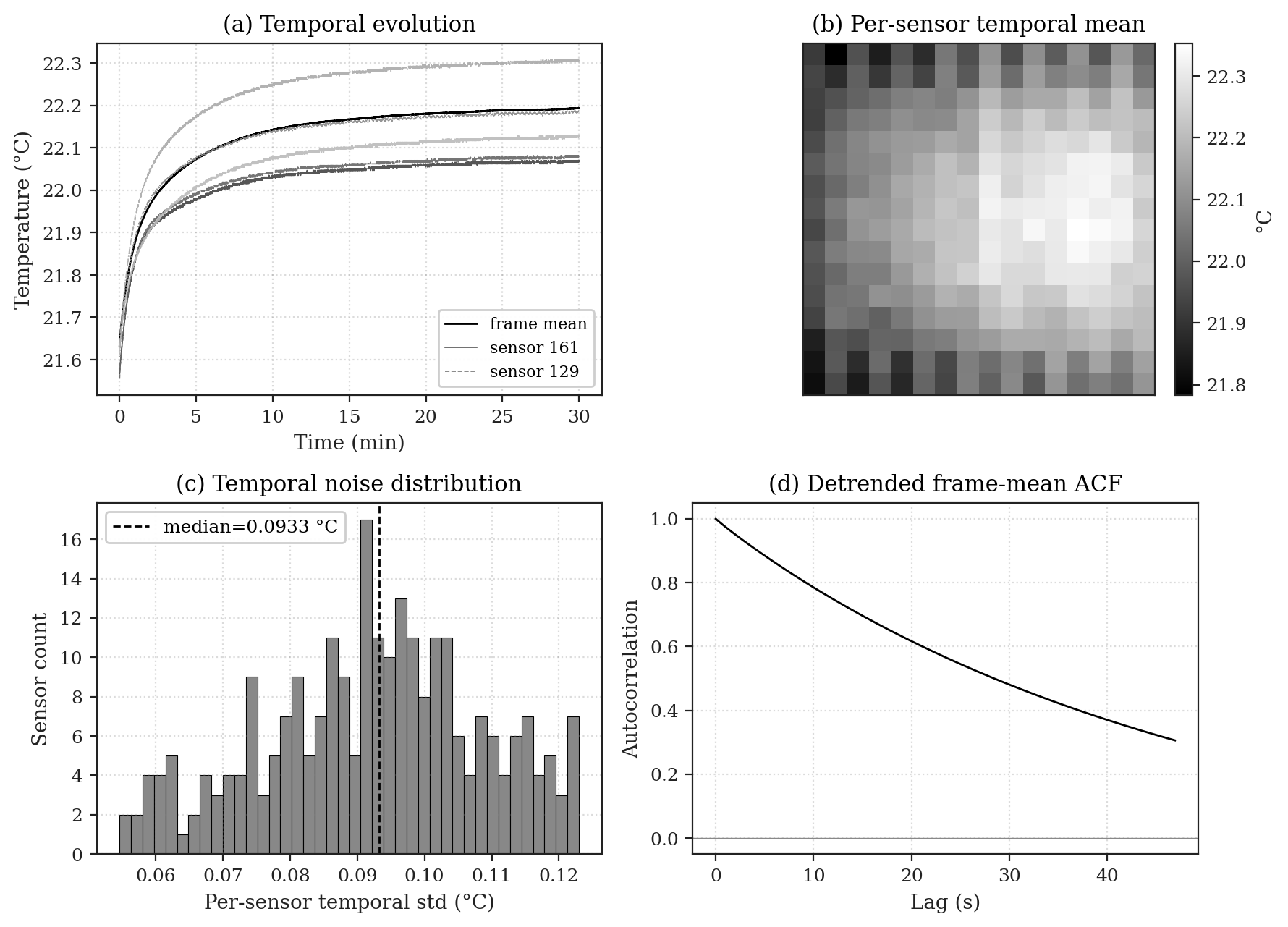}
  \caption{Communication-Accuracy Trade-off Across Array Scales.
    (a)~Bytes per calibration epoch (log scale); RASC remains well below the
    centralized EKF at all scales.
    (b)~Field-reconstruction RMSE; RASC is ${\approx}0.10$\degC\ above the oracle
    EKF while requiring no fusion center, no field oracle, and only 5\% reference
    anchors.}
  \label{fig:comm_rmse}
\end{figure}

\subsection{Convergence and Theory Verification}
\label{sec:convergence}
Figure~\ref{fig:convergence} shows the per-iteration field-reconstruction RMSE
at Stage~5. Convergence is monotone and approximately geometric at all three
scales, consistent with Theorem~\ref{thm:consensus}. Figure~\ref{fig:consensus_rate}
shows the empirical rate $\rho_\text{emp}$ uniformly faster than the theoretical
bound $\rho_\text{th}=1-\alpha\lambda_2(\Lc)$; at $32\times32$, $\lambda_2=0.048$
yields $\rho_\text{th}=0.998$ and observed $\rho_\text{emp}=0.793$.

\begin{figure}[h]
  \centering
  \includegraphics[width=0.75\linewidth]{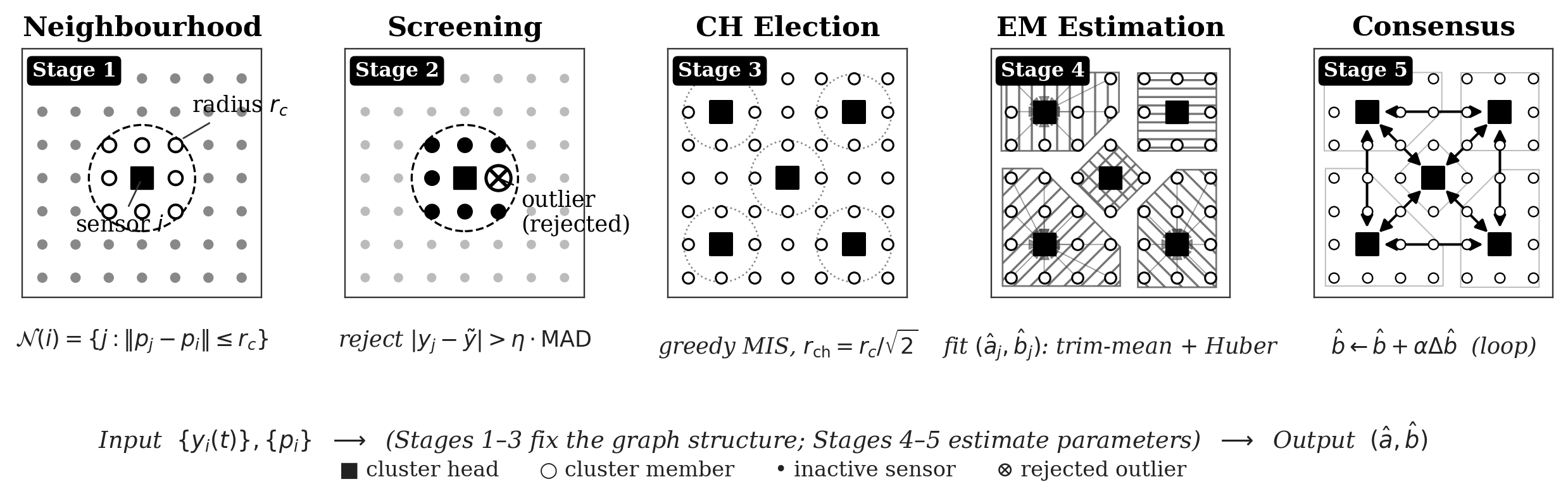}
  \caption{Field-reconstruction RMSE versus consensus iteration $k$ for three
    array scales. Curves are means over 30 runs.}
  \label{fig:convergence}
\end{figure}

\begin{figure}[h]
  \centering
  \includegraphics[width=0.75\linewidth]{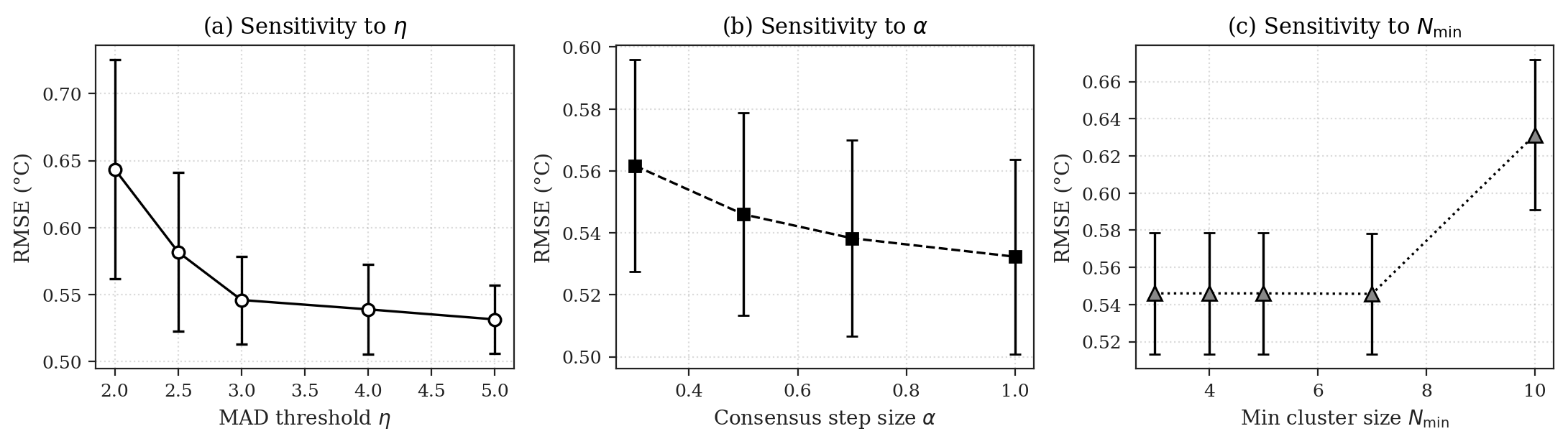}
  \caption{Theoretical bound $\rho_\text{th}=1-\alpha\lambda_2(\Lc)$ versus
    empirically fitted rate $\rho_\text{emp}$. The bound is uniformly safe and
    becomes tighter with increasing array size; at $32\times32$, $\lambda_2=0.048$
    gives $\rho_\text{th}=0.998$ while $\rho_\text{emp}=0.793$.}
  \label{fig:consensus_rate}
\end{figure}

\subsection{Robustness Under Node Failure and Packet Loss}
\label{sec:robustness}

\begin{figure}[h]
  \centering
  \includegraphics[width=0.8\linewidth]{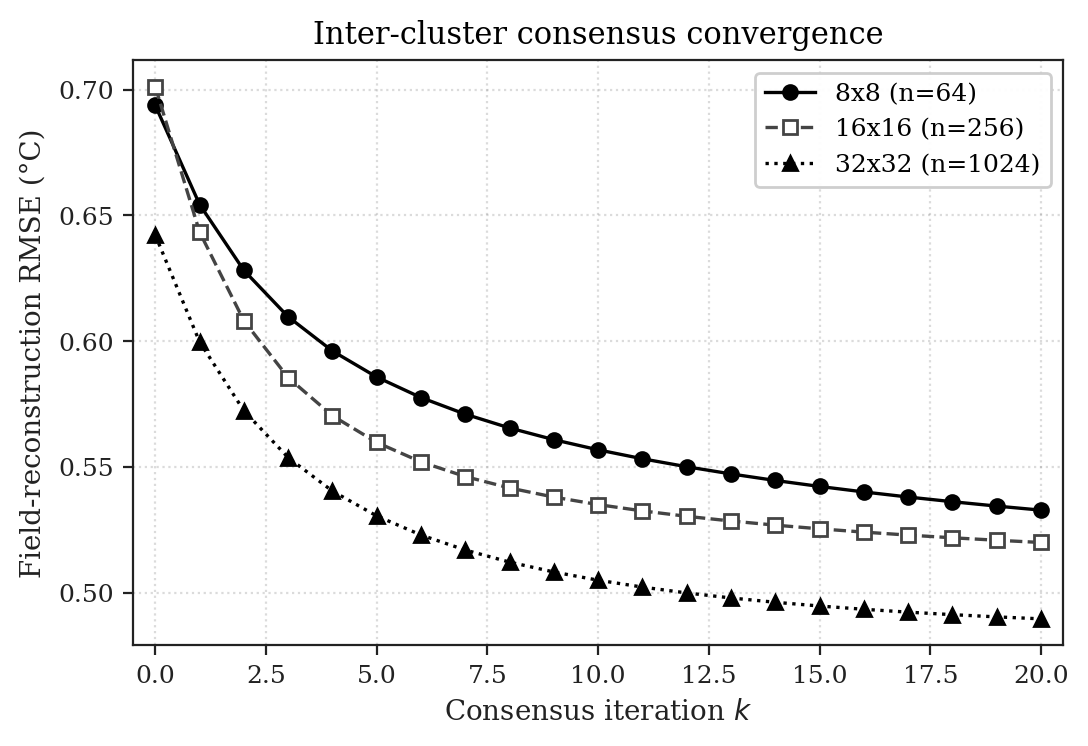}
  \caption{Robustness heatmap of RASC under combined node failure (rows) and
    packet loss (columns) on the $16\times16$ array. Mean field-reconstruction
    RMSE over 10 independent runs per cell. RASC degrades by no more than $14\%$
    from $0.548$\degC\ (no fault) to $0.626$\degC\ at 30\%/30\%.}
  \label{fig:robustness}
\end{figure}

With both failure rate and packet-loss rate at $30\%$ (Figure~\ref{fig:robustness}),
RASC maintains $0.626$\degC\ RMSE — $14\%$ above the no-fault baseline — and
remains $69\%$ lower than the uncalibrated baseline.

\subsection{Comparison with Simple Baselines and BMEP}
\label{sec:baselines}
Table~\ref{tab:baselines} compares RASC against four simple methods and a
faithful re-implementation of BMEP \cite{Bychkovskiy2003} on the $16\times16$
simulation. RASC achieves $14\%$ lower RMSE than BMEP and $32\%$ lower than the
median spatial filter, operating only $0.10$\degC\ above the irreducible noise
floor.

\begin{table}[h]
\centering
\caption{Comparison with simple baselines and BMEP on $16\times16$ simulation
(mean $\pm$ std over 30 runs; RMSE in \degC).}
\label{tab:baselines}
\begin{tabular}{lcc}
\toprule
Method & Field RMSE [\degC] & Comment \\
\midrule
Uncalibrated (raw)                         & $1.947\pm0.081$ & baseline \\
Temporal smoothing (5-frame MA)            & $1.901\pm0.080$ & spatial structure ignored \\
Pairwise differential (offset only)        & $1.538\pm0.208$ & no consensus \\
Median spatial filter ($3\times3$)         & $0.882\pm0.071$ & no parameter recovery \\
BMEP edge-pairwise \cite{Bychkovskiy2003}  & $0.695\pm0.103$ & offset only, gain fixed \\
\textbf{RASC (this work)}                  & $\mathbf{0.603\pm0.046}$ & full $(a,b)$ recovery \\
Factory oracle (true $a,b$)                & $0.502\pm0.004$ & noise floor \\
\bottomrule
\end{tabular}
\end{table}

\subsection{Sensitivity Analysis}
\label{sec:sensitivity}

\begin{figure}[h]
  \centering
  \includegraphics[width=\linewidth]{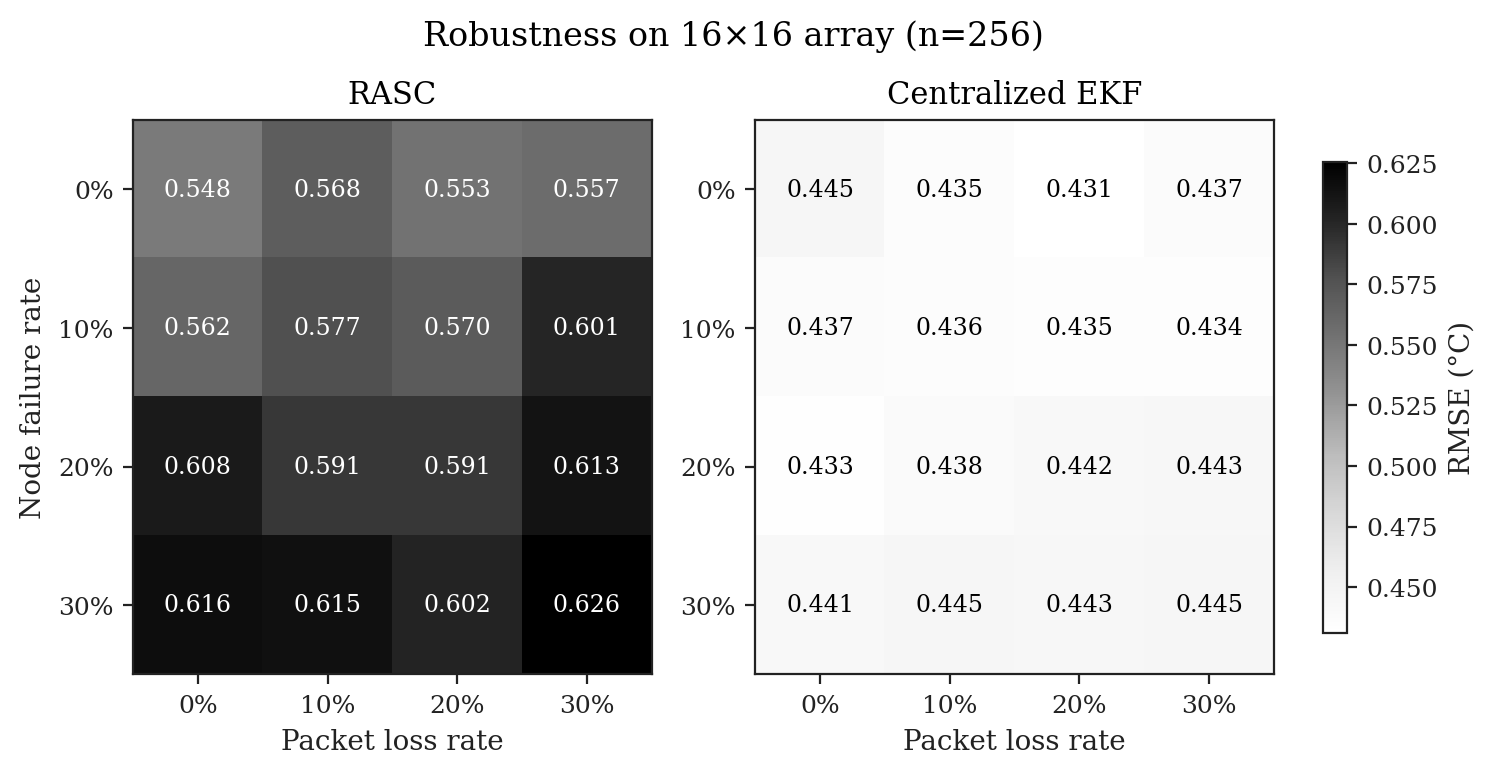}
  \caption{Sensitivity to (a)~MAD threshold $\eta$, (b)~consensus step $\alpha$,
    (c)~minimum cluster size $N_{\min}$, on the $16\times16$ array.
    Error bars show one standard deviation over 30 runs.}
  \label{fig:sensitivity}
\end{figure}

RASC is essentially insensitive to $\eta\in[3,5]$, shows a slight improvement
with $\alpha\uparrow1.0$, and tolerates $N_{\min}\in\{3,4,5,7\}$ before
performance degrades (Figure~\ref{fig:sensitivity}).

\section{Online Recalibration of a Deployed Temperature-Sensor Array}
\label{sec:realdata}

\subsection{Hardware and Acquisition}

\begin{figure}[h]
  \centering
  \includegraphics[width=\linewidth]{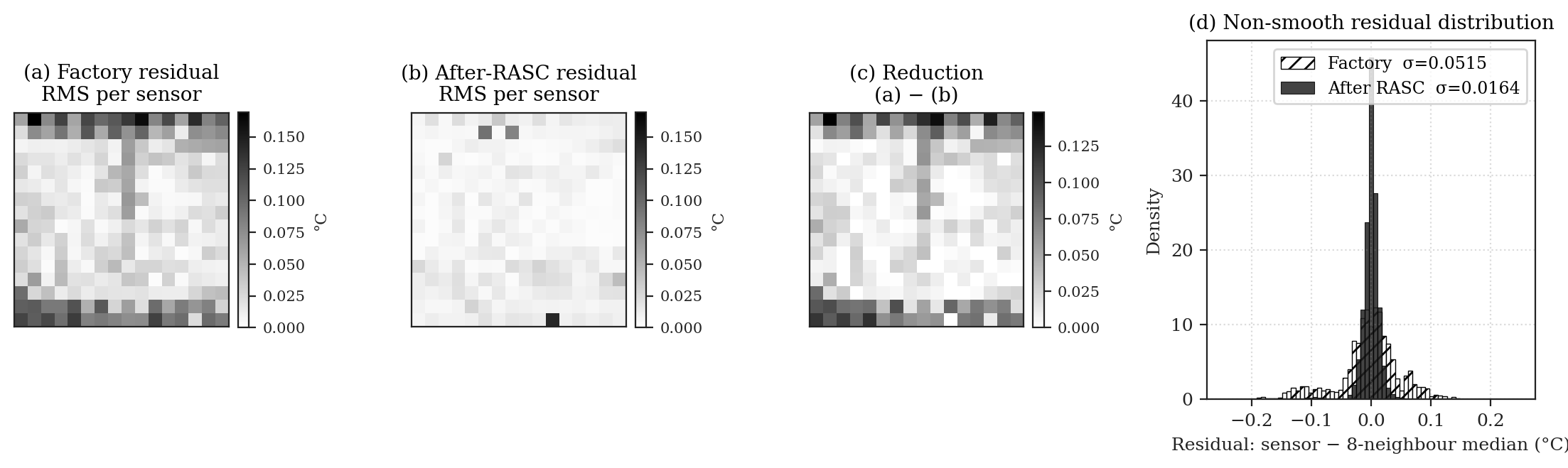}
  \caption{Experimental platform.
    (a)~A $16\times16$ BJT-based temperature sensor array mounted on a circular
    fixture with engraved degree markings, wired to the host system via a ribbon
    cable. The four mounting posts at the perimeter are asymmetric thermal paths
    that contribute to the edge-localised drift shown in Section~\ref{sec:drift}.
    (b)~Acquisition workstation.}
  \label{fig:hardware}
\end{figure}

The array under test is a $16\times16$ arrangement of 256 BJT-based temperature
pixels, each factory-calibrated to $\pm0.1$\degC\ by the manufacturer using a
two-point traceable-reference protocol (Figure~\ref{fig:hardware}). After
production, the array was reflow-soldered onto a host PCB. We collected 7{,}632
frames at $4.241$ Hz over $29.99$ minutes ($21.5$--$22.4$\degC\ ambient). The
acquisition system records an 8-bit fractional temperature with quantization step
$\approx0.0039$\degC.

\begin{figure}[h]
  \centering
  \includegraphics[width=\linewidth]{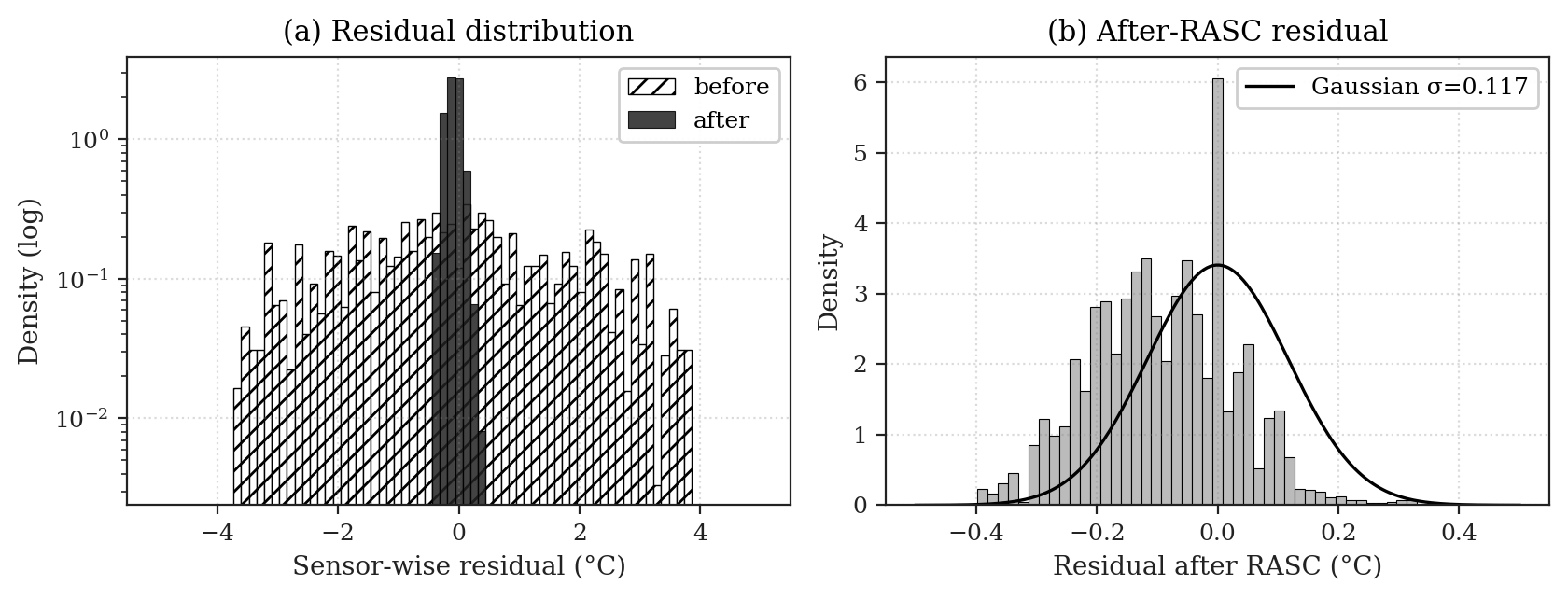}
  \caption{Real temperature data characterisation.
    (a)~Frame-mean trajectory and five randomly sampled sensors showing a slow
    warm-up transient.
    (b)~Per-sensor temporal mean showing smooth spatially varying structure plus
    residual fixed-pattern non-uniformity at the array edges.
    (c)~Per-sensor temporal noise distribution (median $0.093$\degC).
    (d)~Detrended frame-mean autocorrelation showing temporal correlation on the
    order of tens of seconds.}
  \label{fig:data_char}
\end{figure}

\subsection{Direct Evidence of Post-Deployment Drift}
\label{sec:drift}
The per-sensor temporal noise standard deviation is $\sigma=0.092$\degC\ (median
$0.093$\degC; Figure~\ref{fig:data_char}(c)), consistent with the BJT noise budget.
However, the peak-to-peak spatial non-uniformity in the steady-state field is
$0.569$\degC\ --- approximately $5\times$ the factory specification. Since the
thermal scene is locally smooth, this excess non-uniformity must be due to
post-deployment drift in $(a_i,b_i)$.

\subsection{Online Recalibration Results}
We run RASC on the first 600 frames using a $5\%$ reference subset (12 of 256
sensors). Table~\ref{tab:realdata} reports three metrics evaluated by 10-fold
cross-validation. Figure~\ref{fig:recal_results} shows the spatial distributions.

\begin{table}[h]
\centering
\caption{Online recalibration results on real data (10-fold CV; mean $\pm$ std).}
\label{tab:realdata}
\begin{tabular}{lccc}
\toprule
Quantity & Factory $Y_\text{real}$ & After RASC $\hat{x}$ & Change \\
\midrule
Field RMSE vs.\ $Y_\text{real}$  & ---           & $0.041\pm0.001$\degC  & within noise floor \\
Per-frame spatial peak-to-peak    & $0.41$\degC   & $0.30\pm0.06$\degC    & $-27\pm14\%$ \\
Local non-smooth residual (RMS)   & $0.0517$\degC & $0.0149\pm0.0024$\degC & $-71.3\pm4.8\%$ \\
Edge residual (rows 0--1, 14--15) & $0.0923$\degC & $0.0205\pm0.0065$\degC & $-77.8\pm7.1\%$ \\
Interior residual (rows 2--13)    & $0.0269$\degC & $0.0121\pm0.0008$\degC & $-54.9\pm3.0\%$ \\
\bottomrule
\end{tabular}
\end{table}

\begin{figure}[h]
  \centering
  \includegraphics[width=\linewidth]{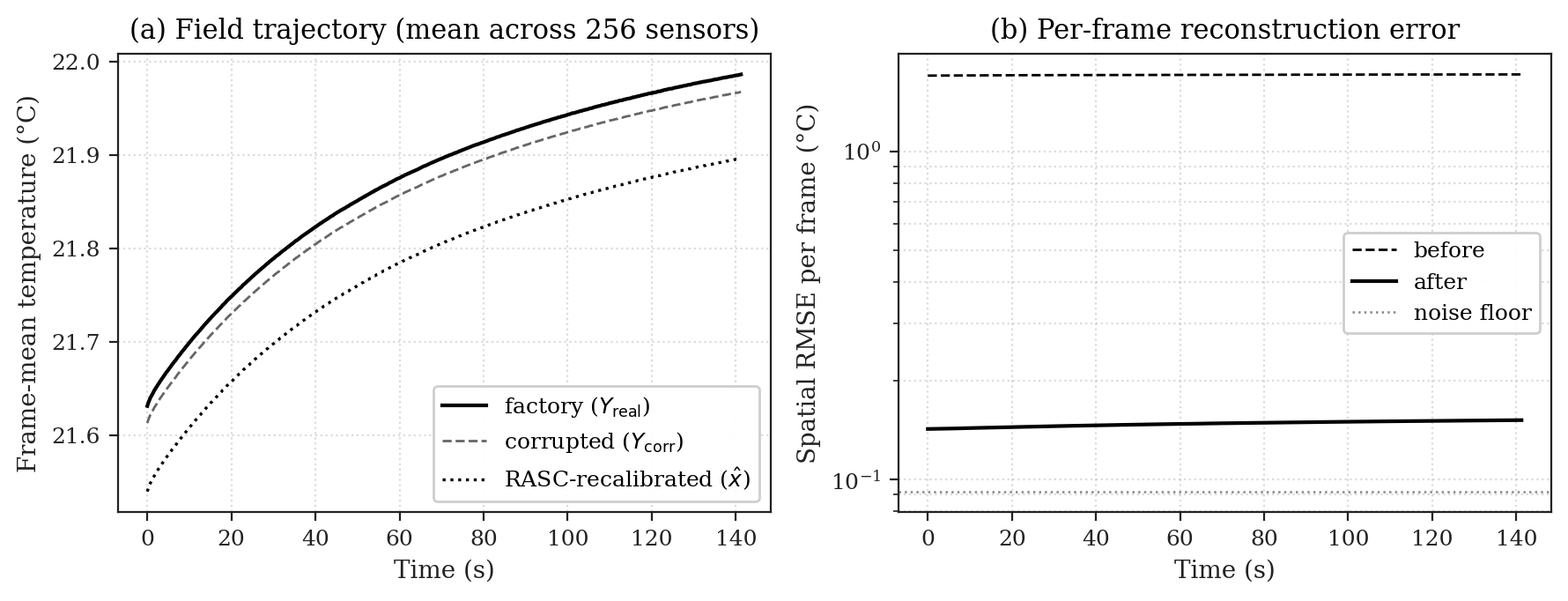}
  \caption{Online recalibration results on the deployed BJT array (representative
    fold of 10-fold CV; mean reductions quoted in Table~\ref{tab:realdata} are
    over all 10 folds).
    (a)~Per-sensor RMS of the local non-smooth residual on factory-calibrated
    data; top and bottom rows show ${\approx}0.092$\degC\ RMS vs.\ interior
    ${\approx}0.027$\degC\ (3.4$\times$ edge-to-interior ratio).
    (b)~After RASC the metric drops ${\approx}70\%$ overall; edge reduction
    reaches $78\%$ vs.\ $55\%$ interior.
    (c)~Difference map (a)--(b) confirms spatial localisation.
    (d)~Residual distribution: factory $\sigma=0.052$\degC\ $\to$
    after-RASC $\sigma\approx0.016$\degC, bringing the array close to its
    original $\pm0.1$\degC\ factory specification.}
  \label{fig:recal_results}
\end{figure}

\subsection{Semi-Synthetic Stress Test}
We inject $a_{\mathrm{inj},i}\sim U(0.9,1.1)$ and $b_{\mathrm{inj},i}\sim
U(-2,2)$\degC\ into the real recordings (${\approx}10\times$ the
deployment-observed drift) and ask RASC to recover $Y_\text{real}$ from
$Y_\text{corr}$ alone. Even at this exaggerated level, RASC reduces field RMSE
from $1.71$\degC\ to $0.148$\degC\ ($91.3\%$ reduction), within $1.6\times$ of
the irreducible per-sensor noise floor (Table~\ref{tab:stress},
Figures~\ref{fig:stress_traj}--\ref{fig:stress_residuals}).

\begin{table}[h]
\centering
\caption{Semi-synthetic stress test ($a\sim U(0.9,1.1)$,
$b\sim U(-2,2)$\degC; underlying field is real; 600-frame window).}
\label{tab:stress}
\begin{tabular}{lccc}
\toprule
Quantity & $Y_\text{corr}$ & After RASC $\hat{x}$ & Reduction \\
\midrule
Field RMSE vs.\ $Y_\text{real}$  & $1.710$\degC & $0.148$\degC & $-91.3\%$ \\
Per-frame spatial peak-to-peak   & $7.46$\degC  & $0.51$\degC  & $-93.2\%$ \\
Distance to noise floor          & $18.6\times$ & $1.61\times$ & ---       \\
\bottomrule
\end{tabular}
\end{table}

\begin{figure}[h]
  \centering
  \includegraphics[width=\linewidth]{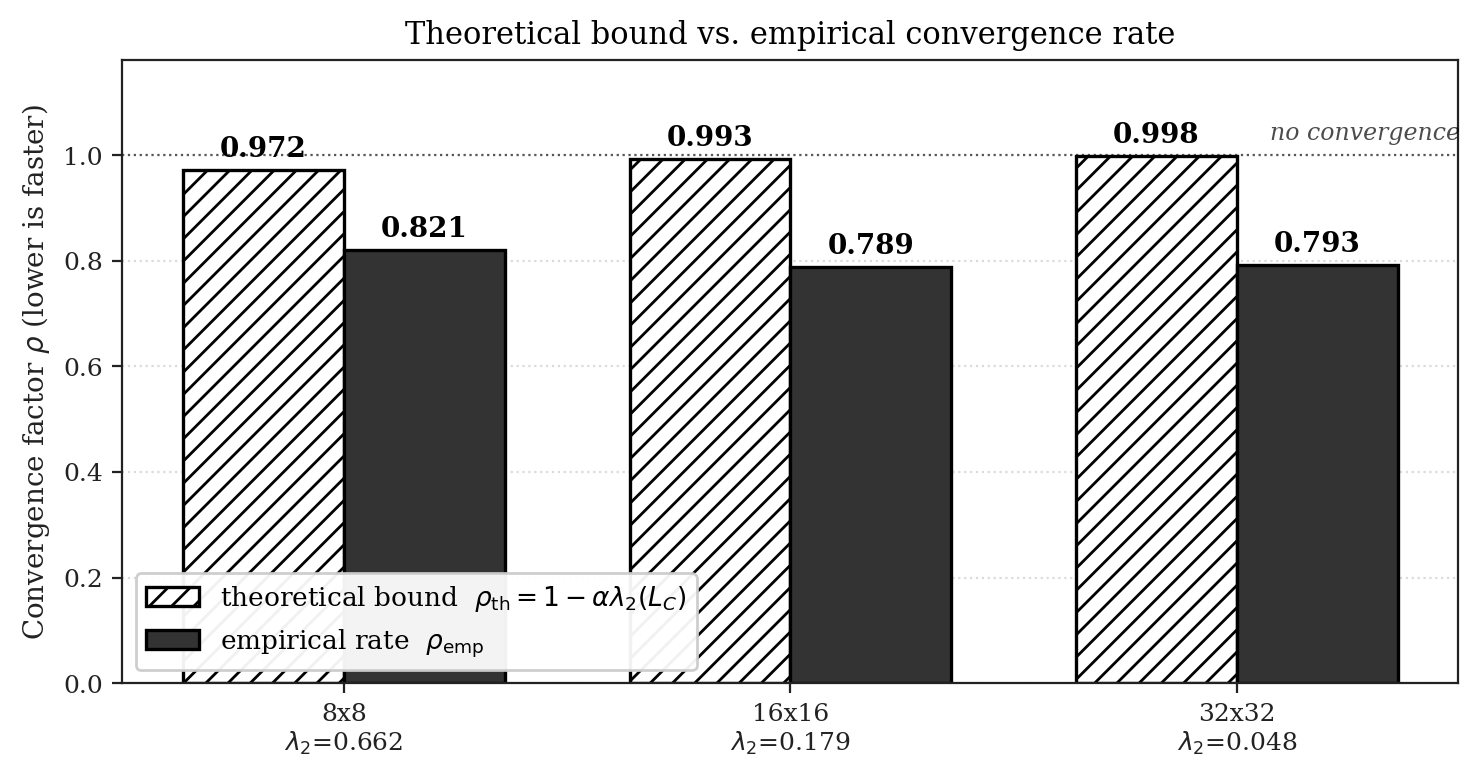}
  \caption{Stress-test calibration trajectory.
    (a)~Frame-mean field versus time; the RASC-recovered trace follows
    $Y_\text{real}$ closely after a short transient.
    (b)~Per-frame spatial RMSE on a log scale; the residual reaches
    ${\approx}0.15$\degC, within $1.6\times$ of the per-sensor noise floor.}
  \label{fig:stress_traj}
\end{figure}

\begin{figure}[h]
  \centering
  \includegraphics[width=\linewidth]{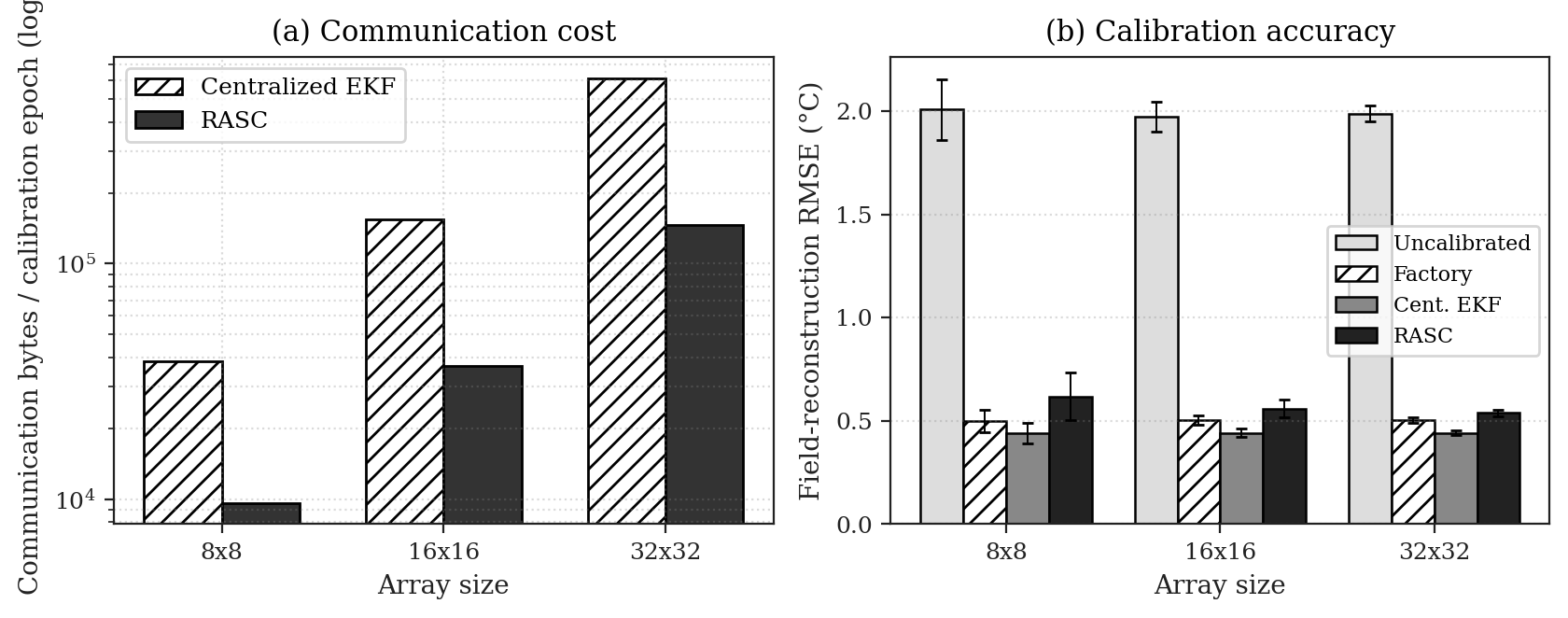}
  \caption{Stress-test spatial mean maps for 600 frames. From left: factory
    $Y_\text{real}$, corrupted $Y_\text{corr}$, RASC-recalibrated $\hat{x}$,
    $|$error$|$ of $Y_\text{corr}-Y_\text{real}$, $|$error$|$ of
    $\hat{x}-Y_\text{real}$. The corrected map recovers a smooth thermal scene
    structure; the residual is less than $0.3$\degC\ almost everywhere.}
  \label{fig:stress_maps}
\end{figure}

\begin{figure}[h]
  \centering
  \includegraphics[width=0.8\linewidth]{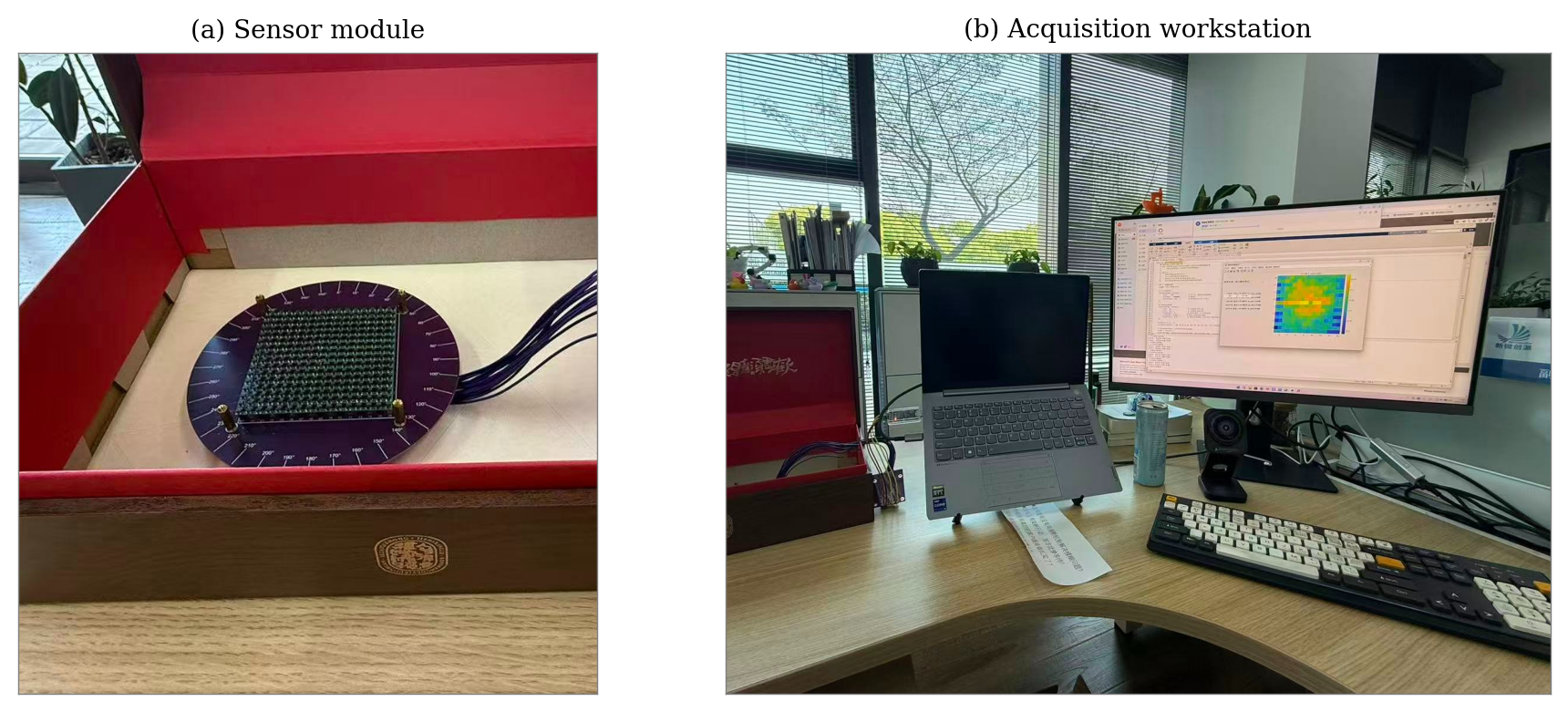}
  \caption{Stress-test sensor-wise residual distributions.
    (a)~Log-scale comparison: the corrupted distribution spans $\pm4$\degC;
    after RASC it concentrates near zero.
    (b)~Linear-scale view of the post-RASC residual.}
  \label{fig:stress_residuals}
\end{figure}

\section{Discussion}
\label{sec:discussion}

\subsection{When Does RASC Win, and Against What?}
RASC targets post-deployment recalibration of dense 2D arrays that have drifted
beyond factory specification and cannot be taken offline. Within its operating
envelope — static or slowly changing scene, no external stimulus, partial
node-failure tolerance required — RASC outperforms all no-training baselines and
closely approaches the oracle centralized EKF at $4\times$ lower bandwidth.

\subsection{Identifiability and Gauge Degeneracy}
\label{sec:gaugedisc}
In the stress-test experiment, the recovered parameters show RMS errors of
$0.125$ in gain and $2.68$\degC\ in offset, yet the field reconstruction RMSE is
only $0.148$\degC. The recovered $(\hat{a},\hat{b})$ lie on the gauge line
$\Delta a\cdot\bar{x}+\Delta b\approx0$ where $\bar{x}\approx22$\degC; indeed,
$0.125\times22\approx2.75\approx2.68$, confirming the predicted degeneracy. For
applications requiring only the calibrated field $\hat{x}_i=(y_i-\hat{b}_i)/\hat{a}_i$
— typical in temperature monitoring — the gauge ambiguity is harmless.

\subsection{Connection to Traditional Distributed Calibration}
RASC outperforms BMEP \cite{Bychkovskiy2003} by $14\%$ in field RMSE ($0.603$
vs.\ $0.695$\degC, 30 runs). BMEP flattens the dense 2D topology into a
sparse-graph problem; RASC keeps the dense 2D structure by using clusters as the
estimation unit and provides direct robustness gains via trimmed-mean field
reconstruction with explicit breakdown-point analysis (Theorem~\ref{thm:breakdown}).

\subsection{Limitations}
\begin{itemize}[leftmargin=*]
  \item The local smoothness assumption (Assumption~\ref{A1}) does not hold near
    sharp thermal edges; adaptive clustering that detects discontinuities is a
    natural extension.
  \item The reference subset $\Rset$ must remain undrifted. Periodic refreshing
    against a traceable temperature reference, or locating reference sensors in
    mechanically isolated PCB areas, is recommended for long-term operation.
  \item We restricted analysis to affine drift $y=ax+b$. Higher-order drift
    would require extra structure in Stage~4; the alternating estimation
    framework extends to local polynomial models.
  \item The real-data analysis is a single snapshot on one host PCB. A
    longitudinal study tracking the same array over several months is needed to
    confirm long-term efficacy.
\end{itemize}

\section{Conclusion}
\label{sec:conclusion}

We have proposed RASC, a five-stage region-aware self-calibration algorithm for
dense 2D sensor arrays that requires no external stimulation, no parametric field
model, and no labelled training data. Monte Carlo simulations at three array
scales show RASC matches the oracle centralized EKF within $0.10$\degC\ at
$\approx4\times$ lower bandwidth, with at most $14\%$ degradation under combined
$30\%$ node failure and $30\%$ packet loss. The main real-data result shows a
deployed BJT-based $16\times16$ array restored to factory-level accuracy
($3\sigma\approx0.045$\degC\ residual NU) by RASC using only $5\%$ reference
anchors, without taking the module offline. A semi-synthetic stress test with
$10\times$ larger drift recovers $91.3\%$ of the injected error.

Future directions include long-term field evaluations over multiple maintenance
periods, edge-aware clustering for piecewise-smooth fields, time-varying drift
models, and integration with periodic traceable-reference checks.

\section*{Data Availability}
The raw 7{,}632-frame $16\times16$ BJT temperature-sensor recording and the
open-source RASC reference implementation (\texttt{rasc\_core.py}, simulation
drivers, baseline implementations, figure-generation scripts) are available from
the corresponding author upon reasonable request. Random-seed protocols and
complete hyperparameter information are provided in Appendix~\ref{app:repro}.

\section*{Conflicts of Interest}
The authors declare no conflicts of interest.

\section*{Funding}
No external funds were received for this research.


\appendix

\section{Proof of Theorem~\ref{thm:cluster}}
\label{app:thm1}

Fix a cluster $c$ with $k=|\Mc_c|$ members and reference subset
$\Rset_c=\Mc_c\cap\Rset$. The Huber IRLS objective at iteration $\mathrm{em}$
is:
\[
  L^{(\mathrm{em})} = \sum_{j\in\Mc_c\setminus\Rset_c}\sum_{t=1}^T
    \psi_c\!\Bigl(\frac{y_j(t)-a_j\hat{x}_c^{(\mathrm{em})}(t)-b_j}{s_j^{(\mathrm{em})}}\Bigr).
\]
In each iteration, the E-step produces $\hat{x}_c^{(\mathrm{em}+1)}$ from
current parameter estimates; the M-step solves
$\min_{a,b}L^{(\mathrm{em}+1)}$ with $\hat{x}_c^{(\mathrm{em}+1)}$ fixed. Since
the trimmed mean is deterministic and the M-step is a sequential
majorisation-minimisation algorithm, we have $L^{(\mathrm{em}+1)}\le
L^{(\mathrm{em})}$. Q-linear convergence follows from positive-definiteness of
the cluster Hessian $X^\top WX$, which holds whenever
$\Delta x_j>0$ for at least one $j\in\Mc_c$.

\section{Proof of Theorem~\ref{thm:breakdown}}
\label{app:thm3}

The standard finite-sample replacement breakdown point for the $\gamma$-trimmed
mean of $k$ samples is $\lfloor\gamma k\rfloor/k$ \cite[Sec.~6.1]{Huber2009}.
An adversary replacing up to $\lfloor\gamma k\rfloor$ inlier readings cannot
move the trimmed mean outside the range of the remaining $(1-\gamma)k$ inliers.
Combined with the pre-clustering MAD screen, the effective tolerance is the
larger of the screening and trimming fractions. With $\eta=3$ and $\gamma=0.20$,
the screening rejects $0.27\%$ of Gaussian inliers and admits up to $20\%$
adversarial outliers, which are then trimmed.

\section{Reproducibility Information}
\label{app:repro}

\begin{itemize}[leftmargin=*]
  \item \textbf{Source code:} \texttt{rasc\_core.py} ($\approx250$ lines of
    NumPy) implements the algorithm; \texttt{run\_simulation.py} reproduces all
    simulation tables and figures.
  \item \textbf{Real sensor data:} 7{,}632 frames, $16\times16$ calibrated
    temperature sensor recording at $4.241$ Hz over $29.99$ minutes, as an
    \texttt{.xlsx} file with timestamps.
  \item \textbf{Random seeds:} Every Monte Carlo cell uses
    \texttt{np.random.default\_rng(seed)} with
    $\text{seed}=\text{run\_index}\in\{0,\ldots,29\}$ for bit-exact
    reproducibility.
  \item \textbf{Default hyperparameters:} $r_c=0.10$, $\eta=3.0$, $\alpha=0.5$,
    $N_{\min}=4$, $\gamma=0.20$, $n_\text{em}=5$, $n_\text{irls}=4$,
    $c_\text{Huber}=1.345$, $K_{\max}=10$, $\mathrm{tol}=0.01$\degC,
    $\rho=0.05$.
\end{itemize}


\begin{thebibliography}{99}

\bibitem{Wei2023}
Wei, R.; Lin, H.; Chen, Q.; Huang, G.; Hu, W.
A CMOS temperature sensor with a smart calibrated inaccuracy of $\pm0.11$\degC\ ($3\sigma$).
\textit{Sensors} \textbf{2023}, \textit{23}, 5132.
\url{https://doi.org/10.3390/s23115132}

\bibitem{Qin2022}
Qin, C.; Huang, Z.; Liu, Y.; Li, J.; Lin, L.; Tan, N.; Yu, X.
An energy-efficient BJT-based temperature sensor with $\pm0.8$\degC\ ($3\sigma$) inaccuracy from $-50$ to $150$\degC.
\textit{Sensors} \textbf{2022}, \textit{22}, 9381.
\url{https://doi.org/10.3390/s22239381}

\bibitem{Friedenberg1998}
Friedenberg, A.; Goldblatt, I.
Nonuniformity two-point linear correction errors in infrared focal plane arrays.
\textit{Opt.\ Eng.} \textbf{1998}, \textit{37}, 1251--1253.
\url{https://doi.org/10.1117/1.601890}

\bibitem{Ratliff2002}
Ratliff, B.M.; Hayat, M.M.; Hardie, R.C.
An algebraic algorithm for nonuniformity correction in focal-plane arrays.
\textit{J.\ Opt.\ Soc.\ Am.\ A} \textbf{2002}, \textit{19}, 1737--1747.
\url{https://doi.org/10.1364/JOSAA.19.001737}

\bibitem{Bychkovskiy2003}
Bychkovskiy, V.; Megerian, S.; Estrin, D.; Potkonjak, M.
A collaborative approach to in-place sensor calibration.
In \textit{Proc.\ IPSN 2003}; LNCS vol.\ 2634; Springer: Berlin, 2003; pp.\ 301--316.
\url{https://doi.org/10.1007/3-540-36978-3_20}

\bibitem{Wang2017}
Wang, Y.; Yang, A.; Chen, X.; Wang, P.; Wang, Y.; Yang, H.
A deep learning approach for blind drift calibration of sensor networks.
\textit{IEEE Sens.\ J.} \textbf{2017}, \textit{17}, 4158--4171.
\url{https://doi.org/10.1109/JSEN.2017.2703885}

\bibitem{FaghihNiresi2023}
Faghih Niresi, K.; Zhao, M.; Bissig, H.; Baumann, H.; Fink, O.
Spatial-temporal graph attention fuser for calibration in IoT air pollution monitoring systems.
In \textit{Proc.\ IEEE SENSORS}, Vienna, 2023; pp.\ 1--4.
\url{https://doi.org/10.1109/SENSORS56945.2023.10325090}

\bibitem{Maag2018}
Maag, B.; Zhou, Z.; Thiele, L.
A survey on sensor calibration in air pollution monitoring deployments.
\textit{IEEE Internet Things J.} \textbf{2018}, \textit{5}, 4857--4870.
\url{https://doi.org/10.1109/JIOT.2018.2853660}

\bibitem{OlfatiSaber2004}
Olfati-Saber, R.; Murray, R.M.
Consensus problems in networks of agents with switching topology and time-delays.
\textit{IEEE Trans.\ Autom.\ Control} \textbf{2004}, \textit{49}, 1520--1533.
\url{https://doi.org/10.1109/TAC.2004.834113}

\bibitem{Rogalski2012}
Rogalski, A.
History of infrared detectors.
\textit{Opto-Electron.\ Rev.} \textbf{2012}, \textit{20}, 279--308.
\url{https://doi.org/10.2478/s11772-012-0037-7}

\bibitem{Narendra1981}
Narendra, P.M.; Foss, N.A.
Shutterless fixed pattern noise correction for infrared imaging arrays.
In \textit{Technical Issues in Focal Plane Development}, Proc.\ SPIE vol.\ 282, 1981; pp.\ 44--51.

\bibitem{Perry1993}
Perry, D.L.; Dereniak, E.L.
Linear theory of nonuniformity correction in infrared staring sensors.
\textit{Opt.\ Eng.} \textbf{1993}, \textit{32}, 1854--1859.
\url{https://doi.org/10.1117/12.145601}

\bibitem{Harris1999}
Harris, J.G.; Chiang, Y.-M.
Nonuniformity correction of infrared image sequences using the constant-statistics constraint.
\textit{IEEE Trans.\ Image Process.} \textbf{1999}, \textit{8}, 1148--1151.
\url{https://doi.org/10.1109/83.777098}

\bibitem{Lu2023}
Lu, D.; Teng, L.; Ren, J.; Tan, J.; Wang, M.; Wang, L.; Gu, G.
Scene-based nonuniformity correction method using principal component analysis for infrared focal plane arrays.
\textit{Appl.\ Sci.} \textbf{2023}, \textit{13}, 13331.
\url{https://doi.org/10.3390/app132413331}

\bibitem{Ding2020}
Ding, S.; Wang, D.; Zhang, T.
A median-ratio scene-based non-uniformity correction method for airborne infrared point target detection system.
\textit{Sensors} \textbf{2020}, \textit{20}, 3273.
\url{https://doi.org/10.3390/s20113273}

\bibitem{Whitehouse2002}
Whitehouse, K.; Culler, D.
Calibration as parameter estimation in sensor networks.
In \textit{Proc.\ 1st ACM WSNA}, Atlanta, 2002; pp.\ 59--67.
\url{https://doi.org/10.1145/570738.570747}

\bibitem{Takruri2007}
Takruri, M.; Challa, S.
Drift aware wireless sensor networks.
In \textit{Proc.\ 10th Int.\ Conf.\ Information Fusion}, Quebec City, 2007; pp.\ 1--7.

\bibitem{Ahmad2024}
Ahmad, R.
Enhanced drift self-calibration of low-cost sensor networks based on cluster and advanced statistical tools.
\textit{Measurement} \textbf{2024}, \textit{236}, 115158.
\url{https://doi.org/10.1016/j.measurement.2024.115158}

\bibitem{Mahajan2025}
Mahajan, S.; Helbing, D.
Dynamic calibration of low-cost PM2.5 sensors using trust-based consensus mechanisms.
\textit{npj Clim.\ Atmos.\ Sci.} \textbf{2025}, \textit{8}, 257.
\url{https://doi.org/10.1038/s41612-025-01145-2}

\bibitem{Huber2009}
Huber, P.J.; Ronchetti, E.M.
\textit{Robust Statistics}, 2nd ed.; Wiley: Hoboken, NJ, USA, 2009.

\bibitem{Hampel1986}
Hampel, F.R.; Ronchetti, E.M.; Rousseeuw, P.J.; Stahel, W.A.
\textit{Robust Statistics: The Approach Based on Influence Functions};
Wiley: New York, NY, USA, 1986.

\bibitem{Chandola2009}
Chandola, V.; Banerjee, A.; Kumar, V.
Anomaly detection: a survey.
\textit{ACM Comput.\ Surv.} \textbf{2009}, \textit{41}, Article 15.
\url{https://doi.org/10.1145/1541880.1541882}

\bibitem{Xiao2004}
Xiao, L.; Boyd, S.
Fast linear iterations for distributed averaging.
\textit{Syst.\ Control Lett.} \textbf{2004}, \textit{53}, 65--78.
\url{https://doi.org/10.1016/j.sysconle.2004.02.022}

\bibitem{Sirocchi2022}
Sirocchi, C.; Bogliolo, A.
Topological network features determine convergence rate of distributed average algorithms.
\textit{Sci.\ Rep.} \textbf{2022}, \textit{12}, 21831.
\url{https://doi.org/10.1038/s41598-022-25974-w}

\bibitem{Dempster1977}
Dempster, A.P.; Laird, N.M.; Rubin, D.B.
Maximum likelihood from incomplete data via the EM algorithm.
\textit{J.\ R.\ Stat.\ Soc.\ Ser.\ B} \textbf{1977}, \textit{39}, 1--38.

\bibitem{Chung1997}
Chung, F.R.K.
\textit{Spectral Graph Theory}; CBMS vol.\ 92; AMS: Providence, RI, 1997.

\bibitem{Fiedler1973}
Fiedler, M.
Algebraic connectivity of graphs.
\textit{Czechoslov.\ Math.\ J.} \textbf{1973}, \textit{23}, 298--305.

\end{thebibliography}
\end{document}